\documentclass[runningheads]{llncs}
\usepackage{graphicx}
\usepackage{cite}
\usepackage{amsmath,amssymb,amsfonts}
\usepackage{algorithmic}
\usepackage{graphicx}
\usepackage{textcomp}
\usepackage{xcolor}
\usepackage{adjustbox}
\usepackage[utf8]{inputenc}
\usepackage[english]{babel}
\usepackage[utf8]{inputenc}
\usepackage{floatflt}
\usepackage{graphics}
\usepackage{wrapfig}
\usepackage{lmodern,blindtext}
\usepackage[10pt]{extsizes}
\usepackage{booktabs}

\begin{document}

\title{Multi-objective Consensus Clustering Framework for Flight Search Recommendation}
\titlerunning{Consensus Clustering Framework for Flight Search Recommendation}
%
\author{Sujoy Chatterjee\inst{1,2} \and Nicolas Pasquier\inst{1} \and Simon Nanty\inst{2}\and Maria A. Zuluaga\inst{2,3}}
\authorrunning{S. Chatterjee et al.}

%
\institute{Universit\'e C\^ote d'Azur, CNRS, I3S, UMR 7271, Sophia Antipolis, France\\
\email{\{sujoy.2611@gmail.com, pasquier@i3s.unice.fr\}}
 \and
Amadeus S.A.S, Sophia Antipolis, France \\
\email{\{simon.nanty@amadeus.com\}}
\and
Data Science Department, EURECOM, France\\
\email{\{zuluaga@eurecom.fr\}}
}
\maketitle              
\begin{abstract}
In the travel industry, online customers book their travel itinerary according to several features, like cost and duration of the travel or the quality of amenities. 
To provide personalized recommendations for travel searches, an appropriate segmentation of customers is required. 
Clustering ensemble approaches were developed to overcome well-known problems of classical clustering approaches, that each rely on a different theoretical model and can thus identify in the data space only clusters corresponding to this model.
Clustering ensemble approaches combine multiple clustering results, each from a different algorithmic configuration, for generating more robust consensus clusters corresponding to agreements between initial clusters.
We present a new clustering ensemble multi-objective optimization-based framework developed for analyzing Amadeus customer search data and improve personalized recommendations.
This framework optimizes diversity in the clustering ensemble search space and automatically determines an appropriate number of clusters without requiring user's input. 
Experimental results compare the efficiency of this approach with other existing approaches on Amadeus customer search data in terms of internal (Adjusted Rand Index) and external (Amadeus business metric) validations.
\keywords{Clustering Ensemble \and Multi-objective optimization.}
\end{abstract}

\section{Introduction}

In the travel industry, the volume of sales mainly relates two factors like click rate (i.e., inquiry rate) and translation of it to conversion rate. 
Multiple factors can hinder the customers from purchasing a ticket in spite of searching for a flight itinerary. 
The click to conversation rate can be increased if certain number of offers can be recommended to a certain set of customers. 
Therefore, personalized recommendation based on similarity of customers can be very effective to improve the business strategy of the travel company. 
The customers can be segregated depending on various features like days of advance flight booking, distance covered during the travel, number of their children. 
Hence, clustering algorithms have a major role in order to segment the customers in a better way. 

Clustering algorithms \cite{jain99} are used in travel context to find sets of customers with similar needs and requirements, and to identify hidden relationships between their search query. 
However, traditional clustering algorithm results depend significantly on the algorithmic configuration used, i.e., the algorithm chosen and its parameterization, and its adequacy to the data space properties.
Hence, clustering the same set of customers with different algorithmic configurations can produce significantly different solutions. 

\begin{wrapfigure}{r}{0.34\textwidth}
  \centering
  \includegraphics[width=0.34\textwidth]{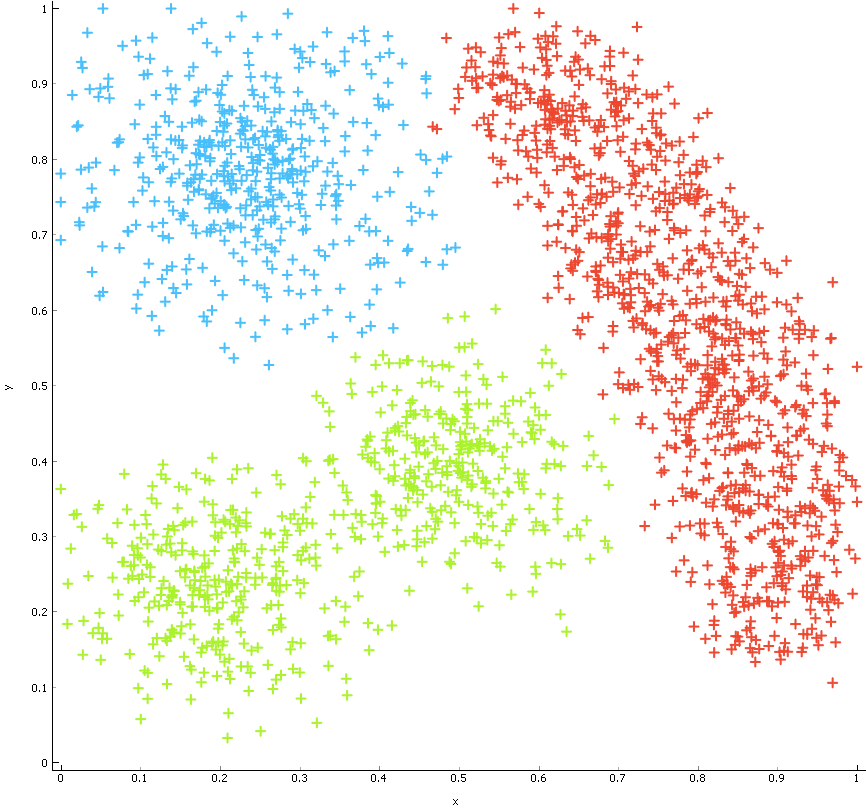}
  \caption{Bi-dimensional data space showing different underlying clustering models.}
  \label{Fig:data_space}
\end{wrapfigure}

Choosing an adequate algorithmic configuration, and defining the final number of clusters as required by most existing clustering approaches, are major practical issues when the required prior knowledge of the data space properties is unavailable.
Consider for example the bi-dimensional data space represented in Fig.~\ref{Fig:data_space} that shows different sub-spaces where groups of objects correspond to different underlying clustering models, e.g., centroid, density or model-based. 
Consensus clustering approaches combine multiple clustering results, obtained from diverse algorithmic configurations, to generate a more robust clustering solution.
By combining the results of algorithms based on different modeling assumptions and for different parameterizations, that is identifying agreements between these \emph{base clusterings} and quantifying the weight of repetitive groups of objects among clusters, these groups of objects can be detected.
Therefore, consensus clustering, that was shown to be an effective approach to generate quality clustering solutions \cite{ayad2008,strehl2002,ZhongYZL15}, is an interesting solution for travel context where little prior knowledge of the evolving data space is available. 
However, to the best of our knowledge, no study on the integration of consensus clustering for better personalized recommendation in travel context was reported in the literature.

In this paper, we study the integration of consensus clustering through a multi-objective optimization process for the clustering of flight search queries. 
This process aims to optimize the selection of flight recommendations that are returned by the Amadeus flight search engine.
For that purpose, a clustering solution is used to segment the space of customers, so that the search engine is optimized independently for each cluster, and customers with different needs and requirements are provided with different recommendations. 
In this context, we have no prior knowledge about the data space modeling assumptions,  like data distribution or natural number of clusters, and choosing an appropriate algorithmic configuration is an important issue.
Using classical clustering approaches requires to apply the different clustering algorithms with different parameterizations multiple times and compute the Amadeus business metric for each clustering solution, implying numerous time consuming computations. 
This metric is the difference between the estimated booking probability of the flight recommendations output by the search engine before and after the optimization of the flight search engine. 
In order to reduce computation cost and optimize the search space exploration, i.e., the potential consensus clusterings, we propose a new ensemble clustering framework.

The clustering ensemble problem is usually posed as an optimization problem where the average similarity of consensus solution with the base clusterings is maximized in order to obtain a better aggregation. 
However, since a consensus solution can be very similar to one base clustering whereas very distant from others, and to remove any kind of bias toward a particular clustering solution, minimizing the standard deviation of these similarity values is also necessary. 
The proposed framework uses a multi-objective clustering ensemble approach that simultaneously optimizes two objective functions \cite{ChatterjeePasquier:2019}:
The maximization of the similarity between the consensus solution and the base clusterings, and the minimization of the standard deviation of these similarities for a consensus solution. 
A formal proof is provided demonstrating that the proposed approach automatically generates a number of clusters at least as appropriate as approaches that consider only co-occurrences of two objects in base clusters for generating a consensus clustering. 
This framework integrates the ensemble clustering solution and a mapping function to categorize a new customer in an appropriate cluster to improve the Amadeus flight search recommendation. 
We present an extensive analysis, by generating diverse sets of base clusterings from multiple perspectives, to demonstrate how the ensemble method performs in terms of consensus solution for customer segmentation as well as improving the Amadeus business metric for better personalized flight recommendation. 
Experimental results with comparative analysis of state-of-the-art methods show that it performs well in majority cases. 
Interestingly, it automatically produces the number of clusters which is close to the best number of cluster produced by other methods that require to specify it each time as input. 

\section{Related Work}
Over the last few years, clustering ensemble has been employed as a useful tool to overcome the drawback of classical clustering algorithms by deriving better clustering solutions by consensus. 
For a given dataset, the clustering ensemble of base clusterings can be produced by applying the same clustering algorithm multiple times with different parameterizations, via sub sampling, or projection of the dataset into different sub-spaces. 
The primary objective of ensemble clustering algorithms is to combine base clustering solutions in such a way that a robust solution is generated to improve the quality of the results compared to base clustering solutions. 
Different approaches have been developed to address this topic \cite{Alqurashi2018,strehl2002,ayad2008,Fred:2002,Liu2014AWS,ZhongYZL15}, and the various state-of-the-art methods can be classified into some major categories described hereafter.

Several approaches consider the clustering ensemble problem as a clustering of categorical data \cite{Nguyen2007}. 
Another category of clustering ensemble methods relies on generating a pair-wise similarity matrix. 
This similarity matrix basically considers the co-association between the objects occurring together in a same cluster for a clustering solution \cite{Fern2003}. 
An alternative method does not rely on an object-object co-association matrix but derives a consensus solution from a cluster association matrix \cite{Mimaroglu2012}. 
Other approaches consider clustering ensemble as a graph, or hypergraph, partitioning problem, and various graph partitioning algorithms were proposed to obtain the consensus solutions \cite{Fern2003,strehl2002}. 
Among graph partitioning based approaches, Strehl and Ghosh modeled this as a hypergraph partitioning problem \cite{strehl2002} and proposed three partitioning approaches: (i) Cluster-Based Similarity Partitioning Algorithm (CSPA), (ii) Hypergraph Partitioning Algorithm (HGPA) and (iii) Meta Clustering Algorithm (MCLA). 
Another recently proposed graph partitioning based novel consensus function namely, weak evidence accumulation clustering (WEAC) and four variants of it \cite{Huang:2015} outperform several other existing baseline approaches. 
The first three variants are basically agglomerative methods, namely average-link (AL), complete-link (CL) and single-link (SL). 
The last one GP-MGLA is a graph partitioning based consensus method. 
However, in all these approaches, the number of clusters is required as input. 
Among other studies, the approach proposed in \cite{Mimaroglu2012} effectively generates a consensus clustering with an automatically defined number of clusters. 
This approach visualizes the base clustering solution as an undirected weighted graph, and Prim's algorithm is adapted to make a minimum-cost spanning tree of the weighted graph. 
Another approach proposed by the same authors also generates automatically the number of clusters, but requires to specify a relaxation parameter \cite{Mimaroglu:2010}. 

\section{Problem Formulation}
Let $X = \{x_1, x_2, \ldots, x_p\}$ denote a set of $p$ customers, where $x_i\in \mathbb{R}^d$, $d$ the number of features used to describe the search query and $Y$ be the set of $n$ clustering algorithms. Here each $x_j$ denotes the customer who searches a query for flight booking. Suppose $C = \{c_1, c_2, \ldots, c_n\}$ be the set of base clustering solutions obtained after applying $n$ clustering algorithms. Each $c_i$ partitions $p$ customers into $k_i$ clusters such that $c_i = \{x_1^{i}, x_2^{i}, \ldots, x_p^{i}\}$, where $c_i$ belongs to the set of all possible partitions of $X$ and $x_j^{i}$ denotes the label of $j^{th}$ customer according to the $i^{th}$ clustering: $x_{j}^i \in \{1, 2,\ldots,k_i\}$ $\forall$ $j \in \{ 1, 2, \ldots, p \}$. Note that, each $c_i$ might comprise  different number of clusters $k_i$. Here, the goal is to derive the best aggregated ensemble solution from the base clustering ones, while automatically determining the number of clusters. In the subsequent sections of the paper we use both the terms customers and objects in an analogous way.

\section{Clustering Ensemble Framework}
The proposed framework integrating consensus clustering optimization in Amadeus flight search engine is presented in Fig.~\ref{Fig:framework}. 
The upper part of the chart shows the creation of the search space, that is the Refined Ensemble Clustering, from the dataset.
The lower part of the chart shows the semi-supervised classification process for learning a customer classification model using consensus clusters multi-objective optimization, combined internal and external validation, characterization of cluster segments and prediction of recommendation class.

\begin{figure}[hbt]
  \centering
  \includegraphics[width=0.9\linewidth]{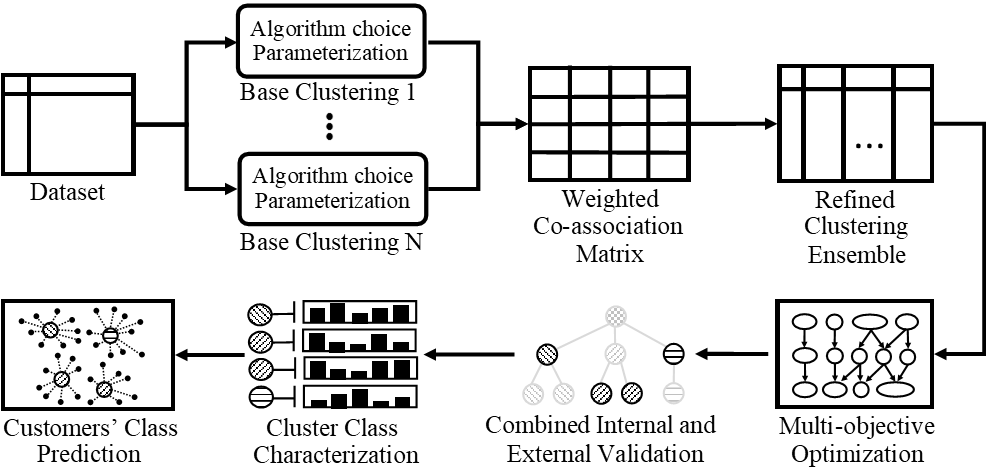}
  \caption{Semi-supervised classification framework for customer segmentation.}
  \label{Fig:framework}
\end{figure}

The central four phases of the proposed ensemble clustering framework are detailed in the following subsections. 
Initially, the base clustering solutions are refined by estimating the maximum number of clusters that can be present in the base clustering. 
To estimate this maximum number of clusters, a weighted co-association matrix \cite{Fred:2002} is constructed and an iterative process is applied on this matrix (Sec~\ref{subsection:matrix}). 
Then, based on the estimated number of clusters, the base clustering solutions are relabeled according to a reference clustering solution (Sec~\ref{subsection:label}), and thus some solutions are refined. 
Ultimately, we combine the original base clustering with the refined set of clustering solutions and thus this ensures the diversity of the clustering solutions. 
In this step, we use a NSGA-II based multi-objective optimization \cite{KDeb2002} to overcome the important complexity of the search for a single solution that is optimal in terms of all the objective functions. 
Hence, in terms of multi-objective optimization, optimality is usually denoted using the concept of Pareto optimality \cite{KDeb2002}. 
A Pareto optimal set basically contains multiple solutions making a trade-off between multiple objectives. 
This NSGA-II method is applied on the set of refined clustering solutions to produce a final set of non-dominated Pareto optimal solutions (Sec~\ref{subsection:multi-objective}). 
Finally, the consensus solution is integrated in the Amadeus application using a mapping function for classifying new customers and make the Amadeus flight search engine return more personalized recommendations (Sec~\ref{subsection:mapping}). 

\subsection{Weighted Co-association Matrix based on Confidence}
\label{subsection:matrix}

In this process, instead of a classical co-association matrix, a weighted co-association matrix is constructed based on different factors like quality of the clustering solution and pair-wise confidence of two objects. 
The quality of a solution is measured by the average similarity of this solution in terms of Adjusted Rand Index (ARI) \cite{hubert:1985} with respect to the other solutions. 
To exemplify the pair-wise confidence of the objects remaining in a same cluster, suppose we have two clustering solutions of 9 objects with respectively 3 and 8 clusters represented as clusterings \{1,1,1,2,2,2,3,3,3\} and \{1,2,3,4,4,5,6,7,8\}. 
Two objects with the same label means that they are assigned to the same cluster in the clustering. 
We can see that the $4^{th}$ and $5^{th}$ objects are in the same cluster in both clustering solutions.  
However, since clustering solution 2 contains a larger number of clusters, the selection criteria for clustering is more restrictive than in solution 1.
This means that the confidence of assigning two objects to the same cluster varies depending on the number of partitions in the clustering solution.
A higher number of partitions means a better confidence in the grouping of objects in the same cluster. 
Therefore, both the quality and the number of clusters of a clustering solution are considered to build the weighted co-association matrix. 
This co-association matrix can be treated as a similarity matrix, where the edge weights depend on these two factors, being calculated on the basis of how many solutions agree to group two particular objects in the same clusters. 
Furthermore, the quality of the solution, in terms of average similarity with respect to the other solutions, and the confidence of two objects being members of the same clusters are also considered. 
In this process, the solutions that contain a higher number of clusters are given more weight than clustering solutions having less number of clusters. 
Besides, quality being a major issue in this purpose, therefore more weight is imposed on the quality metric (i.e., similarity) than the number of components. 
The weight due to the confidence and quality are added to compute the final weight as their similarity.

If $n$ is the number of base clustering and label of each object $j$ is represented by $r_j$, then the similarity $Sim(i,j)$ of two objects $i$ and $j$ is represented by using Eqn. \ref{Eq:coassociation_metric}.

\scriptsize{
\begin{equation}
Sim(i,j) = \sum_{p=1}^{n} (I(r_i=r_j)*cluster(p)) + 2*w \sum_{p=1}^{n} (I(r_i=r_j)*weight(p))
\label{Eq:coassociation_metric}
\end{equation}}
\normalsize

Here, $I$ represents an indicator function and it returns $1$ when the two objects have same labels, otherwise it returns 0. 
In this Eqn~\ref{Eq:coassociation_metric}, $cluster(p)$ is the number of clusters in the $p^{th}$ solution and $weight(p)$ is the weight measured by the similarity of $p^{th}$ solution with other base clustering solutions. 
Due to the difference of two ranges, i.e., variation in number of clusters and similarity value in terms of ARI, $w$ is used to make these two ranges in the same scale. 
To illustrate, consider there are $n$ number of base clustering solutions. 
So, the similarity value of each solution can be computed by comparing it with other $(n-1)$ solutions using ARI metric to obtain $(n-1)$ similarity values.
The average similarity value is calculated from these $(n-1)$ similarity values. 
At the same time, the number of clusters present in each clustering solution can be computed, and we thus measure the average number of clusters that can be present in a clustering solution. 
Now, $w$ is calculated by the ratio between these two values (i.e., ratio of average number of clusters with respect to the average similarity value).

The similarity matrix is then transformed into a graph where objects are vertices and edge weights can be regarded as the strength of bonding between the objects. 
This similarity matrix is then transformed into an adjacency matrix according to a threshold value, and a minimal threshold value is chosen by  $(1/t)$ fraction of maximum edge weight. 
In this step, the value of $t$ should be chosen in such a manner that each object cannot form a separate cluster. 
In our experiments, the value of $t$ was set to 10.
Then, the threshold value is increased by a small amount, and the induced deletion of edges produces another adjacency matrix.
From this adjacency matrix, the strongly connected components are extracted and the number of connected components is observed. 
This task is repeated for a number of time while varying the threshold value in step-wise manner (step size).
This fixed step size is chosen based on sensitivity analysis such that each object cannot produce a separate cluster. 
Therefore, in every step the adjacency matrix is formed and the number of connected components are derived. 
It should be noted that the number of connected components remains fixed for certain threshold values. 
Hence, obtaining the same number of connected components over successive iterations denotes the stability of this particular number of connected components since several clustering solutions agree regarding that specific number of components. 
A sorting over the number of connected components is performed depending on the stability values and the value is chosen from the clustering solutions having the highest similarity value. 
As earlier mentioned, the ARI \cite{hubert:1985} is used to measure the similarity of a clustering solution with respect to the others. 
The final number of clusters cannot be directly selected from the base clustering with highest average similarity since a few clustering solutions can have a high similarity but have a number of connected components without sufficient stability. 
Finally, the initial labeling of the base clusterings are changed according to this estimated number of clusters as described in the following subsections.

\begin{theorem}
The weighted co-association matrix of a clustering solution being calculated from the numbers of co-occurrences of two objects in the same cluster, the confidence of objects co-occurrences, and the quality of the clustering solution, the approach will generate a number of clusters that is at least as close to the number of clusters in the best potential agreement than the number of clusters calculated using a simple co-association matrix where only the numbers of co-occurrences of two objects in clusters are considered.
\end{theorem}

\begin{proof}
Given two clustering solutions $S_p$ and $S_q$ with $m$ and $n$ number of clusters respectively, suppose the confidence of two objects $i$ and $j$ of remaining in same cluster are $Conf^{m}_{S_p}$ and $Conf^{n}_{S_q}$. 
Now according to the proposed model, the values in the co-association matrix not only depend on the number of co-occurrences of two objects being in a same cluster. 
Rather, along with this, the confidence of two objects and the quality of the solutions are also considered. 

Suppose $G$ and $G^{\prime}$ are the graphs constructed from normal co-association matrix (where only count of co-occurrence of two objects lying in same clusters is considered) and weighted co-association matrix (where co-occurrence, confidence and quality of the clustering solutions are considered). 
Accordingly to the proposed model, the confidence of two objects remaining in the same clusters is higher for clusterings with a greater number of partitions than for clusterings with lower number of partitions. 
Therefore, we write $Conf^{m}_{S_p}(i,j) \ge Conf^{n}_{S_q}(i,j)$ where $m > n$ and if $(i,j) \in C^{m}_{S_p}$ then $(i,j) \in C^{n}_{S_q}$. 
Here, $C^{m}_{S_p}$ and $C^{n}_{S_q}$ are two clusters in clustering solution $S_p$ and $S_q$ where the two objects $i$ and $j$ are the members. 

Now, the weighted co-association matrix is constructed based on confidence, co-occurrence and quality of the clustering solution. 
The edges are iteratively removed by increasing repeatedly the threshold $\delta_t$ by a very small amount and observing each time the number of connected components obtained. 
This number of connected components basically denotes the number of clusters. 
Let $N_e(G)$ and $N_e(G^\prime)$ be the number of edges deleted from graphs $G$ and $G^\prime$ each time.

Due to the weight in $G^\prime$ as mentioned previously, it can be written that, $\forall \delta_t$: $N^{\delta_t}_e(G) \ge N^{\delta_t}_e(G^\prime)$, where
$N^{\delta_t}_e(G)$ and $N^{\delta_t}_e(G^\prime)$ denote the number of edges $e$ with same weight $\delta_t$ in graph $G$ and $G^\prime$, respectively. 
Therefore, $\forall \delta_t$: $N^{\delta_t}_c(G) \ge N^{\delta_t}_c(G^\prime)$, where $N^{\delta_t}_c(G)$ and $N^{\delta_t}_c(G^\prime)$ are the number of connected components extracted from $G$ and $G^\prime$ for each value of threshold $\delta_t$.  
We can thus compare the changes in the number of connected components for any two successive iterations in these two graphs, and we have $\frac{d(N_c(G^\prime))}{dt} \le \frac{d(N_c(G))}{dt}$ for any two successive iterations. 
Hence, we obtain a number of clusters that is at least as close to the number of clusters in the best agreement than the number of clusters obtained using a classical co-association matrix, where only the co-occurrences of two objects in a same cluster are considered.
\end{proof}

\subsection{Label Transformation}
\label{subsection:label}

Label transformation aims to unify the labels of common clusters among different clustering solutions.
Let's consider two clustering solutions $c_i =$ $\{x_{1}^{i},$ $x_{2}^{i},$ $\ldots,$ $x_{p}^{i}\}$ and  $c_j = \{x_{1}^{j}, x_{2}^{j},$ $\ldots,$ $x_{p}^{j}\}$ that partition the $p$ customers into $k$ and $m$ clusters respectively. 
Here, each $x_{p}^{i}$ denotes the labeling of the $p^{th}$ object according to $c_i$ clustering solution. 
However, since the cluster labels are symbolic and generated by different processes, there is no correspondence between them in the different clusterings.
In order to represent the clustering solution $c_i$ according to the representation of clusters in solution $c_j$, the number of objects with different labels in $c_j$ and $c_i$ is counted for each label in $c_i$.
Majority voting is then performed among the labels to determine the final corresponding label in $c_i$. 

After estimating the most likely number of clusters, the base clustering solutions containing a higher number of clusters than the estimated value are transformed based on a reference solution containing this estimated number of clusters. 
As mentioned in the last section, the final number of cluster is selected based on the quality of the clustering solutions (highest similarity) of the input clustering. 
Hence, at least one clustering solution with that specific number of cluster should be present. 
To ensure the diversity in the search space, the original base clustering solutions are combined with the refined set of clustering solutions. 
After that, the multi-objective optimization algorithm is applied to derive the non-dominated Pareto optimal consensus solutions.

\subsection{Multi-objective Optimization Algorithm}
\label{subsection:multi-objective}

In this subsection, we outline the utilization of NSGA-II \cite{KDeb2002} with an aim to produce non-dominated near-Pareto optimal solutions. This is explained step-wise below.
\begin{itemize}
\item[$\bullet$] {\bf Encoding Scheme.} Here, the parameters in the search space are represented in the form of a string (i.e., chromosome). Each chromosome represents a clustering solution and chromosomes are encoded with integer value denoting the class label of each object. To exemplify, a chromosome is encoded like $\{r_1, r_2, \ldots, r_n \}$, where $r_i$ represents the class label of $i^{th}$ object. Suppose the encodings of two chromosomes are $\{3,3,2,2,2,1,1,1\}$ and $\{2,2,3,3,3,1,1,1\}$. Here both chromosomes represent the same clustering solutions where objects $\{1,2\}$ are in the single cluster, objects $\{3,4,5\}$ are in another clusters, and objects $\{6,7,8\}$ are in other clusters.

\item[$\bullet$] {\bf Initial Population.} In the initial population, the base clustering solutions obtained after applying different clustering algorithms are taken. In addition to that, the clustering solutions after refinement are included in it. In this way, the diversity of the clustering solutions is maintained. Note that, the number of clusters in each solution of the initial population is not necessary  the same.

\item[$\bullet$] {\bf Selection.} Each chromosome is associated with a fitness function that corresponds to the amount of goodness (fitness value) of the solution encoded in it. The competent chromosomes are selected for further breeding depending on the concept of survival of ``fittest''. In this context, crowded binary tournament is selected as the selection strategy \cite{KDeb2002}.

\item[$\bullet$] {\bf Crossover.} Crossover is a probabilistic procedure to exchange information between two parent chromosomes. In this paper, we use the same crossover operation as described in \cite{chatterjee2013}. Note that, in the clustering ensemble problem, we cannot directly apply the crossover operation. The reason is that two chromosomes can depict the same clustering solution (same fitness value) but with different representations. Then, applying single point/multi point crossover can distort the quality of these solutions. To explain this in more details, suppose there are two chromosomes representing same solutions  $\{2,2,2,1,1,3\}$ and $\{3,3,3,2,2,1\}$. If single point crossover is performed in $3^{rd}$ position, then the new chromosomes become $\{2,2,3,1,1,3\}$  and $\{3,3,2,2,2,1\}$. Hence, the fitness values are decreased although the original solutions were the same. To overcome this limitation we use a bipartite graph based approach described in \cite{chatterjee2013}.

\item[$\bullet$] {\bf Mutation.} In this operation, each chromosome goes through mutation with a slight probability $M_p$. Here, a small float value is added or subtracted to the label of the chromosomes. Note that, as the label of each object is an integer, after the mutation operation the float values obtained are converted into the nearest integer values.
\end{itemize}

In this algorithm, the two following objective functions are simultaneously optimized. 
The first is based on the ARI measure to consider the similarity of a clustering solution with other solutions. 
The second is based on the standard deviation of similarity values to identify potential bias toward a specific clustering solution.
Therefore, the objective functions are maximization of ARI similarity values and minimization of standard deviation among similarity values for a clustering solution with respect to other clustering solutions.


\subsection{Mapping Function}
\label{subsection:mapping}

This application aims to classify new customers in order to customize flight recommendations according to the customer's class.
Since no labeled data exist, the consensus clustering solutions are used in a semi-supervised manner.
Consensus clusters, corresponding each to a customer segment, are characterized to discriminate them, and new customers are classified by assigning them to the cluster they are the most similar to.
The characterization of each cluster is represented as the center point of the cluster from the sample features. 
That is a mean vector of feature values of samples in the cluster. 
The similarity between a new customer and each cluster is then computed based on these vectors, and the customer is assigned to the cluster with maximal similarity.

The mapping function for clustering solution $C^i$ is denoted by $f:x\mapsto \mathrm{argmin}_{\{1 \leq l \leq k \}} (d(\gamma_l,x))$, where $x$ is a new customer, $k$ is the number of clusters in the solution, and $\gamma_l = \sum_{j=1}^p I(l - x^i_j) x_j / \sum_{j=1}^p I(l - x^i_j) $ is the centre of cluster $l$ where $I(x) = 1$ for $x = 0$ and 0 $\forall x \in \mathbb{R}^*$.

\section{Experimental Design \& Results}

In this section, we first detail the Amadeus dataset. 
Then, we discuss the analysis of the proposed model with comparison to other existing models. 
In the experiments, the crossover rate is 0.9, mutation rate is 0.01 and population size is twice the number of the base clustering solutions.

\subsection{Dataset Description and Preprocessing}

To prepare the dataset, we extracted search queries of flight bookings for flights departing from the US during one week on January 2018. 
There are 9 relevant features: Distance between two airports, geography, number of passengers, number of children, advance purchase, stay duration, day of the week of the departure, day of the week of the return, and day of the week of search. 
In ``Geography" the values are 0 for domestic flights (origin and destination are in the same country), 1 for continental flights (origin and destination belong to the same continent), and 2 for intercontinental flights. 
As this dataset contains a very large number of customers (in millions), and as many of them have very similar feature values, the populations are divided into some strata based on similar characteristics. 
Then sampling is performed on these sub-population to generated a stratified sampling of the whole dataset while preserving the distribution properties of the original dataset. 
For example, snapshots of the distribution of the ``distance between airports'' feature values in the original dataset and in sample datasets are shown in Fig. \ref{Fig:Distribution_original_data}. 
Finally, three stratified sample datasets were generated with sizes of 500, 1000 and 1500 samples. 

\begin{figure}[hbt]
 \centering
 \includegraphics[width=.32\textwidth]{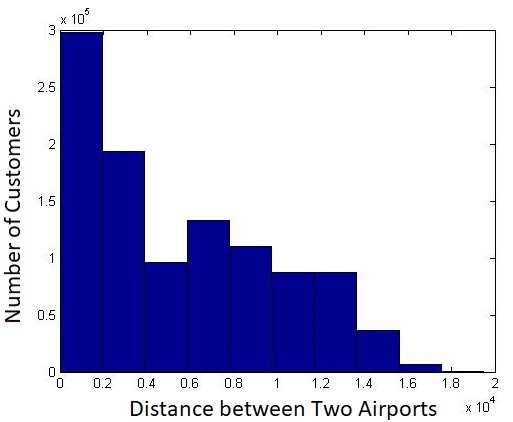}
 \includegraphics[width=.32\textwidth]{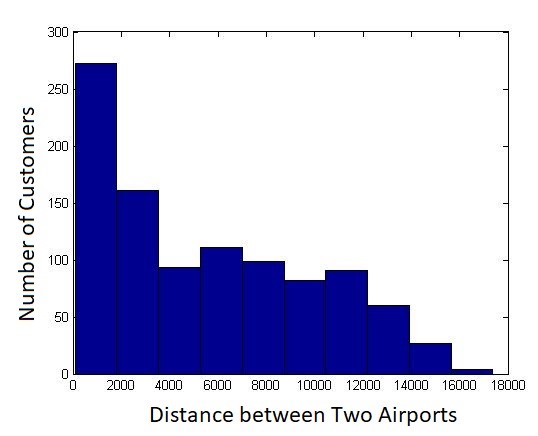}
 \includegraphics[width=.32\textwidth]{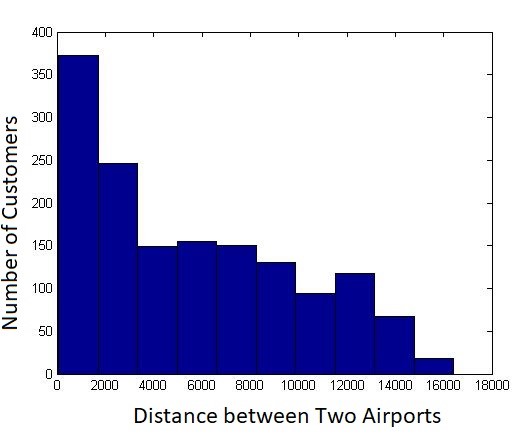}
 \caption{Snapshot of the distribution of feature ``distance between two airports'' for original dataset (Left), sample of size 1000 (Middle) and sample of size 1500 (Right) after stratified sampling.} 
 \label{Fig:Distribution_original_data}
\end{figure}

\subsection{Definition of Base Clusterings}
\label{subsection:base}

After generating the sampled datasets, the $K$-means algorithm was applied, with feature random subspace selection, for different values of $K$ to generate base clusterings with different numbers of clusters. 
Two sets of clustering solutions were generated for each sample-sized dataset. 
As there is no ground truth available for these datasets, the number of clusters parameter $K$ was not kept fixed, but was instead varied between two to eight while applying the clustering algorithms for generating the base clustering solutions. 
In the proposed approach, there is no constraint on the $K$ parameter values that are used to generate base clusterings. 
As described earlier, for each of the three sample-sized datasets, two different sets of base clusterings were generated. 
The first set of base clustering solutions was generated by applying $K$-means ten times, fixing the $K$ parameter for five (or six) of them, and varying the $K$ parameter to generate the remaining five (or four) clustering solutions. 
To generate the second set of base clustering solutions, the value of the $K$ parameter was varied for each $K$-means execution. 

\subsection{Experimental Results}
\label{subsection:experimental_results}

In the first experiment, the accuracy of consensus solutions produced by different methods is measured using the classical ARI metric. 
Since in this context no ground truth is available, internal validation of the solution is performed by comparing it to the base clustering solutions. 
The average similarity of the consensus solution with the base clustering solutions denotes the goodness of the method. 
We made an extensive analysis on different datasets and compared the algorithm with other state-of-the-art methods described in \cite{Huang:2015,strehl2002,Mimaroglu2012}. 
Since the weak evidence accumulation clustering (WEAC) method and four of it variants \cite{Huang:2015} were reported to outperform similar existing other approaches, we choose it as one of the methods for comparison purpose. 
The performance of the proposed model were also compared with the well-known classical methods like CSPA, HGPA, MCLA \cite{strehl2002}. 
If a very limited number of ensemble algorithms automatically estimates the number of clusters, the promising method DiCLENS \cite{Mimaroglu2012} was also used for comparison purposes. 
Specifically, DiCLENS is the most appropriate for comparison purposes as it produces the number of cluster automatically, similarly to the proposed method. 
The results on Amadeus travel dataset for the different samples sizes are given in Tables \ref{table:result_1}-6. 

In Tables \ref{table:result_1}-6, the consensus solution produced by each algorithm is compared with all the base clustering solutions in terms of ARI.
The average similarity obtained is reported in the tables along with the highlighting of the two best scores. 
As mentioned before, we generate two sets of base clustering solutions for each sample (as described in section \ref{subsection:base}). 
It can be realized that the better solution can be achieved if there is a knowledge about the number of cluster or passing $K$ value each time while executing the algorithm. 
In Tables \ref{table:result_1}-3, the result obtained for the first set of base clustering on the three sample datasets are reported and effectiveness of the method is shown. 
In Tables 4-6, the results obtained for the second set of base clustering for each sample, where each clustering solution contains different number of clusters which makes the problem more difficult, are reported. 
It can be observed that for both set of experiments, the proposed method gives consistently good performances even though the number of clusters is not given as input. 
Furthermore, the proposed method produces in several cases the same best quality clustering solution produced by the other methods. 
Also, the number of clusters predicted automatically by the proposed method is very close to the one that generates the best ensemble among all methods. 
Besides, defining the number of clusters for executing the other state-of-the-art approaches (like WEAC, CSPA, MCLA, etc) is extremely difficult because there is no initial knowledge of the value upto which the number of cluster should be varied. 
It is seen that in majority cases the proposed approach outperforms DiCLENS and produces equally good solutions when compared with any other state-of-the-art approaches in terms of ARI measure. 
The non-dominated Pareto front where each point corresponds to the rank-1 solution for different data are demonstrated in Fig. \ref{Fig:plot}. 
In these Pareto front all of the solutions are equally good, however, the final solution is selected based on the highest similarity with the base clusterings.

\begin{table}[hbt!]
\tiny
\parbox{.48\linewidth}{
\centering
\caption{Performance values on the base clustering with 500 samples. Here, six out of ten input clustering solutions contain five clusters and the other solutions contain three, four, six and seven clusters.}
\label{table:result_1}
  \begin{tabular}{|c|c|c|c|c|c|c|}
\hline
Algorithm   & K=3 & K=4  & K=5 & K=6 & K=7 & K=8\\\hline
CSPA  & 0.5919   &   0.5633   & 0.5539  & 0.6472 & 0.5388 & 0.4778  \\\hline
MCLA 	& 0.6117   & 0.7293   &  0.8218  & 0.7217 & 0.8066 & 0.6263 \\\hline
HGPA & 0.2176   & -0.0052   & 0.3447  & 0.2692 & 0.2388 & 0.0089  \\\hline
WEAC-SL  & 0.3972 & 0.6654 & {\bf 0.8275} & 0.8056 & 0.7924 & 0.7770 \\\hline
WEAC-AL & 0.3637 & 0.5964 & {\bf 0.8275} & 0.8066 & 0.7917 & 0.7683\\\hline
WEAC-CL & 0.6001 & 0.6654 & {\bf 0.8275} & 0.8149 & 0.8002 & 0.6913\\\hline
GP-MGLA & 0.6001 & 0.6939 & {\bf 0.8275} & 0.7240 & 0.6995 & 0.6731 \\\hline
DiCLENS & --   & --   & {\bf 0.8275} & -- & -- & -- \\\hline
Proposed & --   & -- & {\bf 0.8275}  & --  & -- & -- \\\hline
\end{tabular}

}
\hfill
\parbox{.48\linewidth}{
\centering
\caption{Performance values on the base clustering with 1000 samples. Here, six out of ten input clustering solutions contain seven clusters and the other solutions contain three, four, five and six clusters.}
   \begin{tabular}{|c|c|c|c|c|c|c|}
\hline
Algorithm   & K=4  & K=5 & K=6 & K=7 & K=8 & K=9\\\hline
CSPA  & 0.5132   &   0.5376   & 0.7162  & 0.7044 & 0.6201 & 0.5814  \\\hline
MCLA 	& 0.6025   & 0.6139   & 0.7822  & 0.8173 & 0.8139 & 0.7455 \\\hline
HGPA & -0.0030   & 0.2010   & 0.3302  & 0.4642 & -0.0048 & -0.0049  \\\hline
WEAC-SL  & 0.4768 & 0.6188 &  0.7490 & {\bf 0.8177} & 0.8140 & 0.8020 \\\hline
WEAC-AL & 0.3353 & 0.5507 &  0.7490 & {\bf 0.8177} & 0.8166 & 0.8043\\\hline
WEAC-CL & 0.6025 & 0.7184 &  0.7490 & {\bf 0.8177} & 0.8166 & 0.7964\\\hline
GP-MGLA & 0.6047 & 0.7184 &  0.7583 & {\bf 0.8177} & 0.7975 & 0.7788 \\\hline
DiCLENS & --   & 0.7183   & -- & -- & -- & -- \\\hline
Proposed & --   & -- &   & --  & {\bf 0.8177} & -- \\\hline
\end{tabular}
}
\end{table}

\begin{table}[hbt!]
\tiny
\parbox{.48\linewidth}{
\centering
\caption{Performance values on the base clustering with 1500 samples. Here, six out of ten input clustering solutions contain five clusters and the other solutions contain three, four, six and seven clusters.}
\label{table:result_2}
 \begin{tabular}{|c|c|c|c|c|c|c|}
\hline
Algorithm  & K=3 & K=4  & K=5 & K=6 & K=7 & K=8\\\hline
CSPA  & 0.4464   &   0.3896   & 0.4580  & 0.4002 & 0.3865 & 0.3504  \\\hline
MCLA 	& 0.4494   & 0.5476   & 0.4962  & 0.4352 & 0.3622 & 0.3589 \\\hline
HGPA & -0.0009   & 0.2474   & 0.3913 & 0.3544 & 0.2761 & 0.2743  \\\hline
WEAC-SL  & 0.4882 & {\bf 0.5584} & 0.5581 & 0.5573 & 0.5557 & 0.5531 \\\hline
WEAC-AL & 0.4049 & {\bf 0.5584} &  0.5567 & 0.5428 & 0.5391 & 0.5308\\\hline
WEAC-CL & 0.4882 & {\bf 0.5584} &  0.5581 & 0.5442 & 0.5359 & 0.4789\\\hline
GP-MGLA & 0.4882 & {\bf 0.5581} & 0.5025 & 0.5009 & 0.4866 & 0.4839 \\\hline
DiCLENS    & -- & {\bf 0.5581} & -- & -- & -- & -- \\\hline
Proposed & --   & {\bf 0.5584} & --  & --  & -- & -- \\\hline
\end{tabular}
}
\hfill
\parbox{.48\linewidth}{
\centering
\caption{Performance values on the base clustering with 500 samples. The 5 clustering solutions each contains different numbers of clusters ranging from 3 to 7.}
   \begin{tabular}{|c|c|c|c|c|c|c|c|}
\hline
Algorithm  & K=2 & K=3 & K=4  & K=5 & K=6 & K=7 \\\hline
CSPA  & -- & 0.5583   &  0.5644   & 0.5410  & 0.6214 & 0.5422  \\\hline
MCLA 	& --   & 0.5841  &  0.7088  & 0.6480 & 0.7263 & 0.5462  \\\hline
HGPA & 0.2176   & 0.3487   & -0.0054  & 0.1360 & 0.6188 & 0.4892   \\\hline
WEAC-SL  & 0.1847 & 0.3689 & 0.6283 & 0.6166 & {\bf 0.7425} & 0.7291 \\\hline
WEAC-AL & 0.0991 & 0.2950 &  0.5152 & {\bf 0.7525} & 0.7263 & 0.7211 \\\hline
WEAC-CL & 0.4638 & 0.5919 &  0.6945 & {\bf 0.7525} & 0.7263 & 0.7163\\\hline
GP-MGLA & 0.4638 & 0.5947 & 0.7088 & {\bf 0.7525} & 0.7263 & 0.7113 \\\hline
DiCLENS & 0.1847   & --   & -- & -- & -- & --\\\hline
Proposed & --   & -- & --  & 0.7378 & -- & --  \\\hline
\end{tabular}
}
\end{table}

\begin{table}[hbt!]
\tiny
\parbox{.48\linewidth}{
\centering
\caption{Performance values on the base clustering with 1000 samples. The 5 clustering solutions each contains different numbers of clusters ranging from 5 to 9.}
\label{table:result_3}
\begin{tabular}{|c|c|c|c|c|c|c|}
\hline
Algorithm    & K=4  & K=5 & K=6 & K=7 & K=8 & K=9\\\hline
CSPA  & 0.5278   &  0.5472  & 0.7094  & 0.6695 & 0.5985 & 0.5967  \\\hline
MCLA 	& 0.6497   & 0.6879  &  0.7772  & 0.6712 & 0.7502 & 0.6941 \\\hline
HGPA & -0.0030   & 0.3478   & 0.3814  & 0.5099 & -0.0047 & -0.0047  \\\hline
WEAC-SL  & 0.6038 & 0.6787 & 0.7722 & 0.7713 & 0.7863 & 0.7810\\\hline
WEAC-AL & 0.5247 & 0.6910 &  0.7716 & 0.7713 & 0.7711 & 0.7861\\\hline
WEAC-CL & 0.6485 & 0.7651 &  0.7722 &  0.7883 & {\bf 0.7882} & 0.7861\\\hline
GP-MGLA & 0.6484 & 0.7683 & 0.7722 & {\bf 0.7865} & 0.7724 & 0.7555 \\\hline
DiCLENS & --   & 0.7683   & -- & -- & -- & -- \\\hline
Proposed & --   & -- & --  & {\bf 0.7865}  & -- & -- \\\hline
\end{tabular}
}
\hfill
\parbox{.48\linewidth}{
\centering
\caption{Performance values on the base clustering with 1500 samples. The 5 clustering solutions each contains different number of clusters ranging from 3 to 7.}
  \begin{tabular}{|c|c|c|c|c|c|c|}
\hline
Algorithm  & K=3 & K=4  & K=5 & K=6 & K=7 \\\hline
CSPA  & 0.3978   &   0.4553   & 0.4985  & 0.4916 & 0.4262   \\\hline
MCLA 	& 0.4863   & 0.5438   &  0.5203 & 0.4957 & 0.3452  \\\hline
HGPA & 0.1735   & -0.0011   & -0.0015 & 0.3053 & 0.3305  \\\hline
WEAC-SL  & 0.2460 & 0.3661 & 0.4973 & 0.5516 & 0.5527  \\\hline
WEAC-AL & 0.1776 & 0.3371 & {\bf 0.5736} & 0.5714 & 0.5589 \\\hline
WEAC-CL & 0.5130 & 0.5522 & {\bf 0.5736} & 0.5714 & 0.5589 \\\hline
GP-MGLA & 0.5186 & 0.5389 & {\bf 0.5698} & 0.5615 & 0.5508\\\hline
DiCLENS    & -- & -- & 0.5581 & -- & --   \\\hline
Proposed & --  & -- & --  & 0.5553 & --    \\\hline
\end{tabular}
}
\end{table}

\begin{figure*}[hbt]
 \centering
 \includegraphics[width=0.45\textwidth]{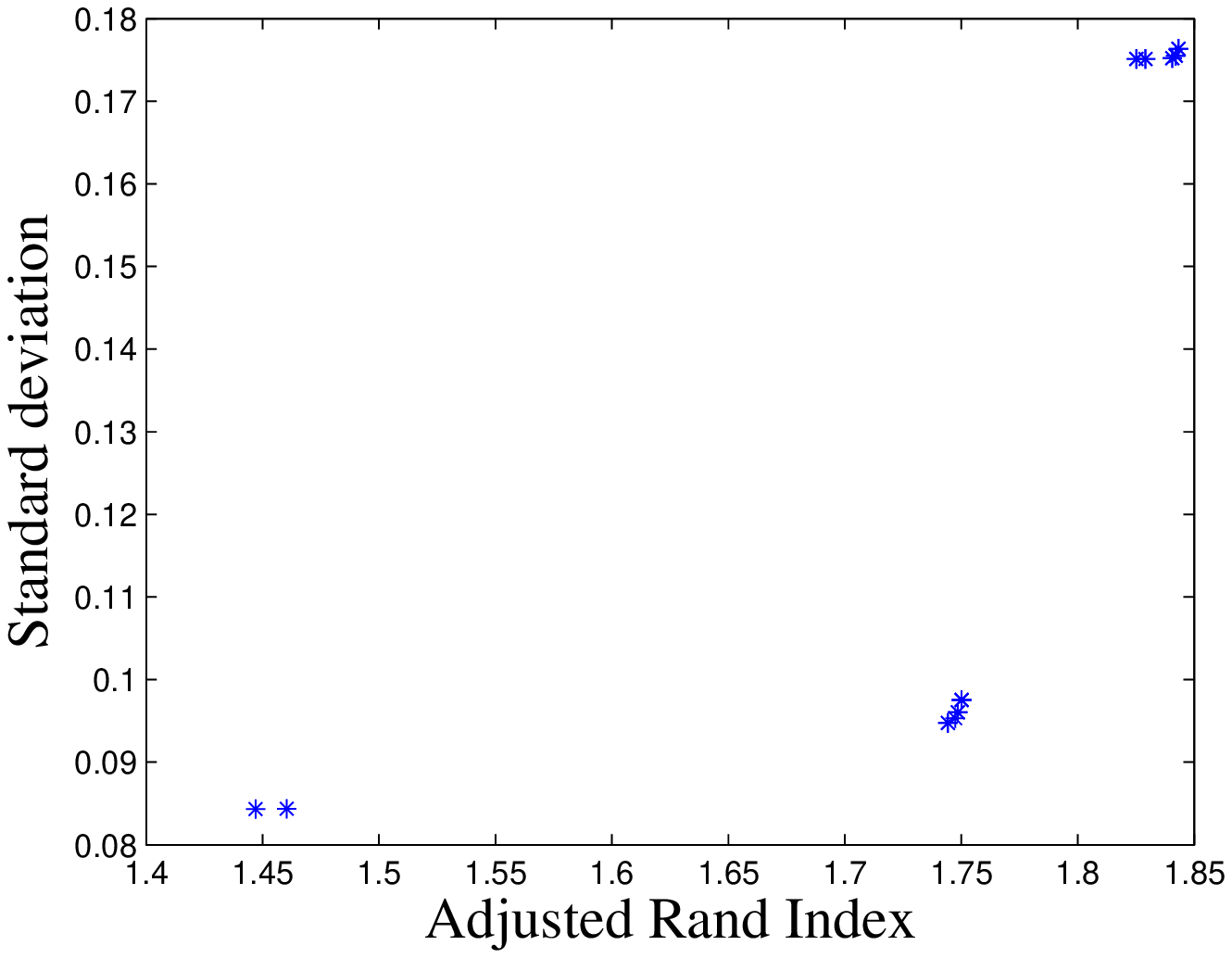} 
 \includegraphics[width=0.45\textwidth]{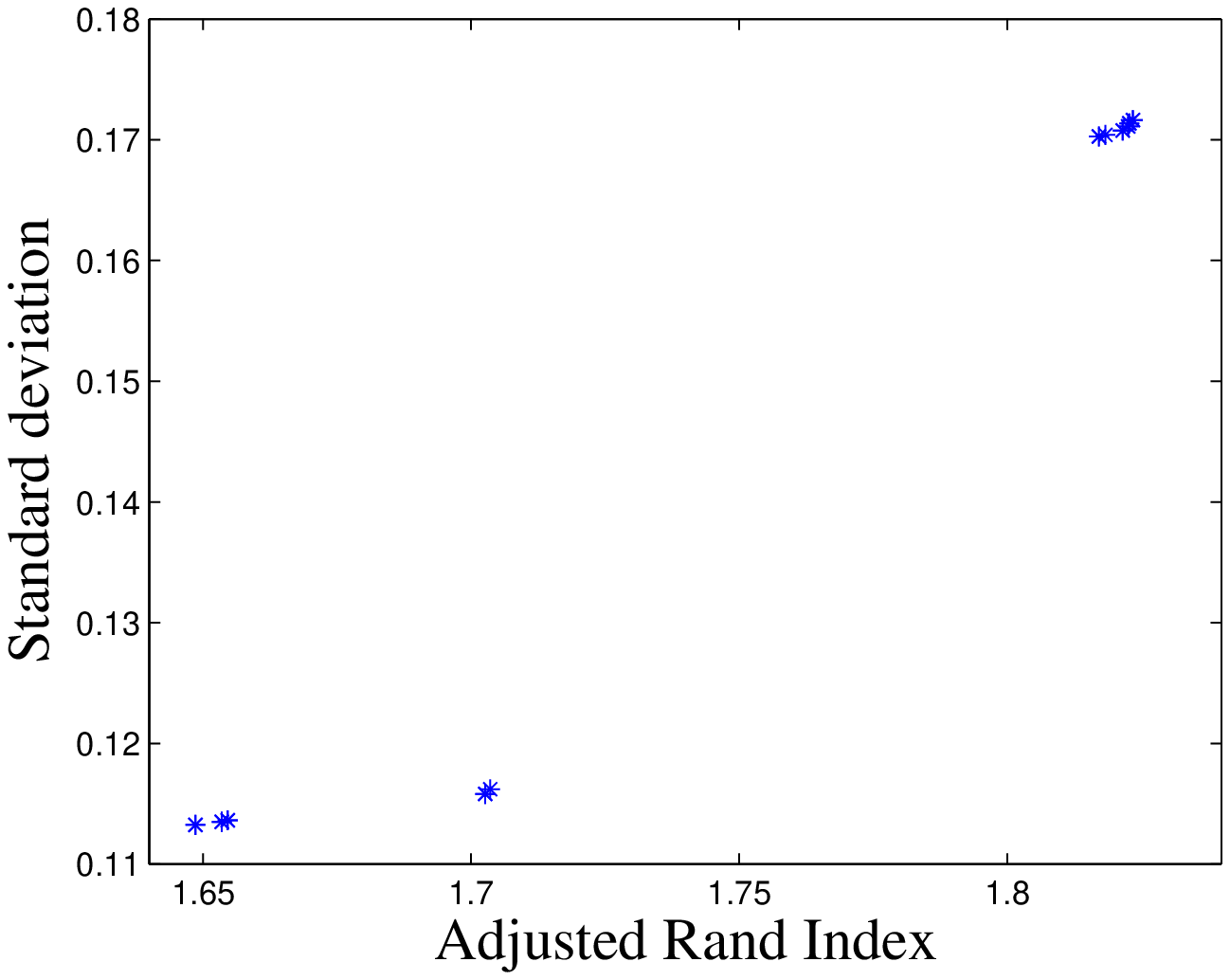}      
 \caption{Pareto front for Amadeus data of sample size 500 (Left) and 1500 (Right).}
 \label{Fig:plot}
\end{figure*}

In this work, consensus clustering techniques are integrated in the flight recommendation selection optimizer of the Amadeus flight search engine. 
Consensus clustering solutions are used for customer segmentation with the objective to optimize the selection strategy according to the diverse categories of users and their different needs. 
The flight selection depends on a quantity defined as a linear combination of different criteria, including the recommendation price and diverse convenience criteria of the flights. 
The linear combination of weights is optimized to maximize the booking probability of the returned recommendations. 
This booking probability is estimated using a Customer Choice Model \cite{lheritier2018} and a mapping function is necessary to assign a new customer to a particular cluster. 
If the mapping function described in Sect. \ref{subsection:mapping} was used, the $K$-Nearest Neighbors approach was also tested to predict the label of a new customer, by finding the $K$-closest customers of this new customer and performing a majority voting over their labels. 
However, considering only $K$ neighbors for predicting the label seems to induce some information loss, and therefore a better accuracy can be obtained using the method described in Sect. \ref{subsection:mapping}.

Experiments presented in Table~\ref{table:result_Amadeus} were conducted on the first set of base clustering solutions, along with the consensus solutions for 500 customers, to perform the optimization process.
The performances were then evaluated according to the Amadeus business metric used during this optimization process: The relative difference between the sum of all the booking probabilities of the flight recommendations returned by the optimized solution and by the reference solution. 
This reference solution is defined by setting all weights to zero except for the recommendation price, and it corresponds to the default configuration of the flight search engine. 
The Amadeus business metric indicates to which extent the optimized solution improves the attractiveness of the recommendations returned by the search engine. 
The obtained percentage reported in the table represents how much the proposed clustering technique improves the internal objective function used to select the flight recommendations in the flight search engine. 
Although there is not a direct link between this improvement and the conversion rate, this percentage represents a surrogate measure to it: The difference of conversion rate induced by the new configuration.

\begin{table}[hbt]
\tiny
\centering
\caption{Performance measure on booking probability improvement in terms of Amadeus business metric.}
\begin{tabular}{|c|c|c|c|c|c|c|}
\hline
Algorithm  & K=3 & K=4  & K=5 & K=6 & K=7 \\\hline
Base
clusterings & 49\% & 8.9\% & 24.4\%,21.6\%,12.2\%,4.4\%,21.5\%,18.6\% & 21.6\% & 21.7\% \\ \hline
CSPA  & 29.7\%   &  14\%   & 27.6\%  & 36.5\% & 19.4\%   \\\hline
MCLA 	& 22.6\%   & 19.7\%   &  12.7\% & 13.3\% & 27.7\%  \\\hline
HGPA & 31.3\%   & 13.5\%   & 31.9\% & 37.3\% & 22.5\% \\\hline
WEAC-SL  & 6.9\% & 36.8\% & 13.2\% & 11.6\% & 21.0\%  \\\hline
WEAC-AL & 22.7\% & 31.2\% & 19.8\% & 21.2\% & 12.8\% \\\hline
WEAC-CL & 20.8\% & 28.6\% & 32.2\% & 36.1\% & 26.6\% \\\hline
GP-MGLA & 16.8\% & 19.7\% & 30.0\% & 27.6\% & 29.6\% \\\hline
DiCLENS    & -- & -- & 28.6\% & -- & --   \\\hline
Proposed & --  & -- & 23.6\%  & -- & --    \\\hline
\end{tabular}
\label{table:result_Amadeus}
\end{table}

The proposed consensus clustering solution gives a better average improvement than most of the base clustering solutions. 
It is a good compromise among consensus solutions as it constantly gives good ARI values and its business metric is over the median of all other consensus methods. Additionally, it saves time compared to the current process in which we compare $N$ base clustering solutions based on the result of the optimization process, the processing time can be divided by $N$.
This is an important feature since the optimization process is the bottleneck part of the application.

It can be seen that some consensus algorithms such as HGPA or CSPA, give higher improvements in term of Amadeus business metric than the proposed method for some $K$ values. 
However, as shown in Tables~1-6, HGPA has a very low ARI, which indicates that it failed to combine the base clustering solutions. 
As it deviates significantly from base clustering solutions, we cannot rely on its solution, and a similar reasoning is applicable to other algorithms. 
In the current Amadeus process, retrieving the business metric for one clustering solution is time consuming, and it is not feasible to compute it for all consensus solutions before selecting one of them, as we did in this study. 
Therefore, we need to choose a reliable algorithm showing acceptable results in terms of both ARI and business metric. 
During the baseline experiments assuming some prior business knowledge about the main features, an improvement of 23.3\% was obtained, which is equivalent to the improvement for the proposed model that does not use any prior knowledge. 
We assume the solution is composed of 6 segments (Business domestic, Business international, Week-end domestic, Week-end international, Others domestic and Others international). 
Furthermore, this assumption that is not data-dependent may not be applicable for all search query datasets, depending on the market, the time period, etc. 
Hence, it is more reliable to depend upon multiple diverse clustering solutions and an appropriate consensus generation process.

\section{Conclusion}
In the travel industry, identifying segments of customers that have close needs and requirements is a key step for generating better personalized recommendation.
We propose a multi-objective optimization based consensus clustering framework to improve customer segmentation and provide better personalized recommendation in the Amadeus flight search engine.
This framework aims to overcome the issues encountered when the segregation of customers relies on a single clustering algorithm that is based on modeling assumptions that do not match, partly or entirely, with the data space regarding the number of clusters, distributions, etc. 
In the context of the selection of flight recommendations returned by Amadeus flight search engine, the proposed framework hold some properties required to generate relevant consensus clustering such as demonstrated by the theoretical proof of its adequate estimation of the number of clusters. 
This consensus clustering based solution was integrated in the Amadeus flight search engine, and its capability to generate better personalized recommendation, while reducing calls to the time consuming part of current Amadeus process, were demonstrated.
The efficiency of the proposed approach regarding application objectives and performances was also shown throughout experiments conducted on Amadeus customer search query data to compare it with other existing approaches. 
As a future direction, we intend to study how other objective functions can be deployed in order to obtain a better clustering solution aiming to improve booking probability.

 \bibliographystyle{splncs04}

\bibliography{New}
%
%
%
%
%
\end{document}